\documentclass[letterpaper, 10 pt, conference]{ieeeconf}
\IEEEoverridecommandlockouts
\renewcommand{\baselinestretch}{0.97}

\usepackage{cite}
\usepackage{amsmath,amssymb,amsfonts}

\usepackage{amsthm}
\usepackage{amsthm}
\usepackage{tikz}
\usetikzlibrary{arrows.meta,positioning,calc,fit,backgrounds,shapes.misc}

\newtheorem{definition}{Definition}[section]

\makeatletter
\let\olddefinition\definition
\let\oldenddefinition\enddefinition
\renewenvironment{definition}[1][]
  {\olddefinition[#1]\itshape}
  {\oldenddefinition}
\makeatother

\usepackage{algorithmic}
\usepackage{algorithm}

\usepackage{graphicx}
\usepackage{textcomp}
\usepackage{xcolor}

\usepackage{bbm}
\usepackage{optidef}
\usepackage{pgfplots}
\pgfplotsset{compat=1.18}

\usepackage{hyperref}
\usepackage[nameinlink,noabbrev]{cleveref}

\theoremstyle{plain}
\newtheorem{theorem}{Theorem}[section]

\newtheorem{lemma}{Lemma}[section]

\newtheorem{problem}{Problem}[section]
\newtheorem{exmp}{Example}[section]

\newtheorem{assumption}{Assumption}[section]

\theoremstyle{remark}
\newtheorem{remark}[theorem]{Remark}

\usepackage{bbm}

\newcommand{\E}{\mathbb{E}}

\def\BibTeX{{\rm B\kern-.05em{\sc i\kern-.025em b}\kern-.08em
    T\kern-.1667em\lower.7ex\hbox{E}\kern-.125emX}}
    
\title{ \LARGE \bf Mean-Field Control of Adherence in Participation-Coupled Vehicle Rebalancing Systems
}
\author{Avalpreet Singh Brar$^{1}$, Rong Su$^{1}$, Jaskaranveer Kaur$^{2}$, Gioele Zardini$^{3}$  % <-this % stops a space
\thanks{This work was supported by the National Research Foundation, Singapore through its Medium Sized Center for Advanced Robotics Technology Innovation (CARTIN) under Project WP2.7.}
\thanks{$^{1}$School of Electrical and Electronic Engineering, Nanyang Technological University, Singapore. (brar0002@e.ntu.edu.sg, rsu@ntu.edu.sg).}
\thanks{$^{2}$Aumovio SE, Singapore ({jaskaranveer.kaur@aumovio.com}).}
\thanks{$^{3}$Laboratory for Information and Decision Systems, Massachusetts Institute of Technology, Cambridge, MA, USA. (gzardini@mit.edu).}
}

\begin{document}

\maketitle

\begin{abstract}
Human driver participation is a critical source of uncertainty in Mobility-on-Demand (MoD) rebalancing.
Drivers follow platform recommendations probabilistically, and their willingness to comply evolves with experienced outcomes. 
This creates a closed-loop feedback in which stronger recommendations increase participation, participation increases congestion, congestion lowers allocation success, and realized allocations update adherence beliefs. 
We propose a microscopic stochastic model that couples (i) belief-driven participation, (ii) Poisson demand, (iii) uniform matching, and (iv) Beta--Bernoulli belief updates. 
Under a large-population closure, we derive a deterministic mean-field recursion for the population adherence state under platform actuation.
For i.i.d.\ Poisson demand and constant recommendation intensity, we prove global well-posedness and invariance of the recursion, establish equilibrium existence, provide uniqueness conditions, and show global convergence in the regime where platform recommendations are no weaker than baseline participation. 
We then define steady-state adherence and throughput, characterize the induced performance frontier, and show that adherence and throughput cannot, in general, be simultaneously maximized under uniform time-invariant actuation. 
This yields a throughput-maximization problem with an adherence floor. 
Exploiting the monotone frontier structure, we show the optimal uniform time-invariant policy is the maximal feasible recommendation intensity and provide an efficient bisection-based algorithm.
\end{abstract}


\section{Introduction}
Vehicle rebalancing addresses the persistent spatial and temporal mismatch between vehicle supply and passenger demand in Mobility-on-Demand (MoD) systems \cite{ZardiniAnnRev2022,cassandras2025control}. Macroscopic dynamical models capture the aggregate evolution of passenger queues and vehicle distributions, enabling equilibrium and stability analysis of ride-hailing systems \cite{guo2021robust, nilsson2023equilibrium}. To support scalable implementation in large urban networks, hierarchical control architectures decompose repositioning decisions across spatial layers \cite{beojone2023two}.

More recent research has explored learning-based methods to address the computational challenges of large-scale fleet control \cite{brar2021dynamic}. Reinforcement learning and graph-based approaches can learn scalable real-time policies that coordinate vehicle movements using the observed system state \cite{singhal2024real}. Other work has investigated the joint design of rebalancing with related operational decisions, including charging operations \cite{singh2024multi, brar2022supply}, dynamic pricing \cite{li2025learning}, and compatibility-based matching mechanisms \cite{brar2025maximal}. In parallel, broader objectives such as fairness and equity in mobility services have also been incorporated into routing and control frameworks \cite{11423297,bai2025routing,brar2020ensuring}. 

Another stream of work examines the strategic interactions that arise in multi-agent mobility markets. Game-theoretic formulations analyze the interaction between centrally controlled fleets and other traffic participants \cite{paparella2025analysis}, while market-oriented approaches study how incentives and pricing mechanisms influence fleet behavior \cite{brar2024integrated}. Examples include receding-horizon game models for ride-hailing markets \cite{maljkovic2024receding} and learning-based methods for computing Stackelberg equilibria in incentive design problems \cite{maddux2025no}. Together, these studies demonstrate that modern rebalancing frameworks increasingly combine optimization, learning, and economic mechanisms.

A key limitation of many rebalancing formulations is that decisions are executed by human drivers whose participation is uncertain. Drivers follow platform recommendations probabilistically based on their confidence and past outcomes. Recent work therefore models adherence explicitly by treating drivers as probabilistic decision makers whose willingness to follow recommendations evolves with realized allocations \cite{brar2024vehicle,lindstroem2025taxi}. In such systems, stronger recommendations increase participation, higher participation increases competition for demand, reduced allocation success lowers confidence, and realized allocations update adherence beliefs. This participation–congestion feedback creates a closed-loop dynamical system in which both adherence evolution and steady-state system performance depend on the recommendation policy.

The main contributions of this paper are as follows.
First, we formulate a stochastic microscopic model of vehicle rebalancing under adherence uncertainty that captures participation, congestion, and belief updates from realized allocations, and obtain a deterministic mean-field recursion for the population-level adherence dynamics under platform recommendations.
Second, we analyze the mean-field recursion, establishing global well-posedness and convergence properties of the adherence dynamics, and characterize equilibrium behavior under constant control.
Finally, we characterize the steady-state trade-off between adherence and throughput and show that these objectives cannot be simultaneously maximized under a uniform time-invariant control policy, leading to a design problem of selecting the largest recommendation intensity satisfying a desired adherence level, for which we propose an efficient bisection-based algorithm.

The remainder of the paper is organized as follows: \Cref{sec:model} presents the stochastic and mean-field formulations, \Cref{sec:theory} develops the theoretical analysis and control design results, and \Cref{sec:numerical} provides numerical simulations.

\section{Model Formulation} \label{sec:model}
We consider a discrete-time stochastic system indexed by $t\in\mathbb{Z}_{\ge 0}$ with a population of $K$ agents (drivers), defined on a filtered probability space $(\Omega,\mathcal{F},\{\mathcal{F}_t\}_{t\ge0},\mathbb{P})$.

At the beginning of epoch $t$, the platform selects a recommendation policy $\{u_i(t) | i\in \mathcal{I}\}$, where,  $\mathcal{I} = \{1,2,...,K\},$ based on past information $\mathcal{F}_{t-1}$.

\subsection{Platform Recommendation Policy}

The platform determines participation recommendations at each epoch.

\begin{definition}[Platform recommendation probability]
\label{def:platform_u}
For each agent $i \in \mathcal{I}$ and epoch $t$, let 
$u_i(t) \in [0,1]$ denote the probability that the platform issues a participation recommendation to agent $i$. The control process $\{u_i(t) | i\in \mathcal{I}\}$ is assumed to be $\mathcal{F}_{t-1}$-measurable.
\end{definition}

After that, agents decide whether to participate. With booked exogenous demand $D(t)$, real-time matching is performed among active agents. The final allocation outcomes $\{A_i(t)| i\in \mathcal{I}\}$ are observed by all agents, and relevant believes will be updated.

\subsection{Probability Space and Information Structure}
For each epoch~$t$, define the random vector
\begin{equation*}
\begin{split}
    Y(t):=\bigl(B_1(t),&\dots,B_K(t),\; D(t),\; A_1(t),\dots,A_K(t)\bigr)\\
    &\in
\Omega_t:=\{0,1\}^K\times\mathbb{Z}_{\ge0}\times\{0,1\}^K,
\end{split}
\end{equation*}
including, in order, participation indicators, demand, and allocation outcomes, where $B_i(t)=1$ if and only if driver $i$ participates at time $t$, and $A_i(t)=1$ if and only if driver $i$ gets allocated at $t$.

Let $\Omega:=\prod_{t=0}^{\infty}\Omega_t$ and write a sample point as $\omega=(\omega_0,\omega_1,\dots)$ with $\omega_t=Y(t)$.
We equip $\Omega$ with the canonical product $\sigma$-algebra $\mathcal{F}:=\sigma(\cup_{t\ge0}\mathcal{F}_t)$ generated by the natural filtration
\[
\mathcal{F}_t:=\sigma\bigl(Y(0),Y(1),\dots,Y(t)\bigr).
\]
All control actions are required to be nonanticipative: the control process $\{u_i(t) | i\in \mathcal{I}\}$ is $\mathcal{F}_{t-1}$-measurable.

\subsection{Agent Belief State and Adherence}
Each agent~$i$ maintains a Beta-parameterized belief (posterior) about the probability that participating under platform coordination yields a successful allocation~\cite{josang2002beta}.

\begin{definition}[Belief state]
For each agent $i\in\{1,\dots,K\}$, let $(\alpha_i(t),\beta_i(t))\in\mathbb{R}_{>0}^2$
denote Beta parameters at time $t$, and define the \emph{belief state} (posterior mean); i.e., the trust of driver $i$ in the recommendation as:
\begin{equation}\label{eq:belief_mean}
x_i(t):=\frac{\alpha_i(t)}{\alpha_i(t)+\beta_i(t)}\in[0,1].
\end{equation}
\end{definition}

Successful allocations increase $\alpha_i$, and unsuccessful participation increase $\beta_i$.

\begin{assumption}[Belief-driven adherence]\label{assumption:exposure}
At epoch $t$, conditional on $\mathcal{F}_{t-1}$, agent $i$ is \emph{adherent} to the platform
(within-epoch commitment to follow recommendations) with probability $x_i(t)$.
These adherence commitments are conditionally independent across agents given $\mathcal{F}_{t-1}$~\cite{brar2024vehicle}.
\end{assumption}

\begin{remark}[Measurability]
Since $(\alpha_i(t),\beta_i(t))$ are updated using information revealed at epoch $t-1$, the belief $x_i(t)$ is $\mathcal{F}_{t-1}$-measurable.    
\end{remark}

\subsection{Participation Decision}

Let $p_i\in[0,1]$ denote the \emph{baseline} participation probability of agent $i$ when the agent is not adherent (i.e., not committing to follow platform recommendations).

Under Assumption~\ref{assumption:exposure} and Definition~\ref{def:platform_u},
the effective participation probability is the mixture
\begin{equation}\label{eq:qi_def}
q_i(t):=(1-x_i(t))\,p_i + x_i(t)\,u_i(t).
\end{equation}

\begin{assumption}[Conditional independence of participation]\label{ass:cond_indep_particip}
Conditional on $\mathcal{F}_{t-1}$, the participation indicators
$\{B_i(t) | i\in \mathcal{I}\}$ are independent and satisfy
\[
B_i(t)\mid \mathcal{F}_{t-1}\sim\mathrm{Bernoulli}(q_i(t)).
\]
\end{assumption}

\subsection{Endogenous Congestion}

For a tagged agent $i$, define the (endogenous) congestion as the number of \emph{other} agents who participate at epoch $t$:
\begin{equation}\label{eq:congestion_def}
M_{-i}(t):=\sum_{j\neq i} B_j(t).
\end{equation}
By Assumption~\ref{ass:cond_indep_particip}, conditional on $\mathcal{F}_{t-1}$,
\begin{equation}\label{eq:poisson_binomial}
M_{-i}(t)\sim\mathrm{PoissonBinomial}\big(\{q_j(t)\}_{j\neq i}\big).
\end{equation}
Moreover, because $B_i(t)$ is conditionally independent of $\{B_j(t)\}_{j\neq i}$ given $\mathcal{F}_{t-1}$,
\begin{equation}\label{eq:M_indep_Bi}
\mathbb{P}\bigl(M_{-i}(t)=k\mid B_i(t)=1,\mathcal{F}_{t-1}\bigr)
=
\mathbb{P}\bigl(M_{-i}(t)=k\mid \mathcal{F}_{t-1}\bigr)
\end{equation}

\subsection{Demand, Matching, and Allocation}

Let $D(t)\in\mathbb{Z}_{\ge0}$ denote exogenous passenger demand at epoch $t$.

\begin{assumption}[Exogenous demand]\label{ass:exo_demand}
The demand process $\{D(t)\}_{t\ge0}$ is independent of $\mathcal{F}_{t-1}$ and of the within-epoch participation decisions at time $t$, and satisfies
\[
D(t)\sim\mathrm{Poisson}(\lambda(t)),
\]
for a (possibly time-varying) deterministic rate $\lambda(t)>0$.
\end{assumption}

Let $N(t):=\sum_{j=1}^K B_j(t)$ denote the number of active agents.
We assume a capacity-$1$ uniform matching mechanism: given $(D(t),N(t))=(d,n)$, exactly $\min\{d,n\}$ active agents are selected uniformly at random (without replacement) and allocated one passenger each.

For a tagged agent $i$ with $B_i(t)=1$ and $M_{-i}(t)=k$ (so $N(t)=k+1$), the conditional allocation probability equals
\begin{equation}\label{eq:pi_def}
\begin{aligned}
\pi(d,k):&=\mathbb{P}\bigl(A_i(t)=1\mid D(t)=d,M_{-i}(t)=k,B_i(t)=1\bigr)\\
&=\min\!\left(1,\frac{d}{k+1}\right).
\end{aligned}
\end{equation}

\begin{definition}[Conditional allocation probability]\label{def:si_def}
The allocation probability of agent $i$, conditional on $B_i(t)=1$ and $\mathcal{F}_{t-1}$, is
\begin{equation}\label{eq:si_sum}
s_i(t)
:=
\mathbb{E}\!\left[\,
\pi\bigl(D(t),M_{-i}(t)\bigr)\ \bigm|\ B_i(t)=1,\mathcal{F}_{t-1}
\right].
\end{equation}
Equivalently, using \eqref{eq:M_indep_Bi} and the law of total expectation,
\begin{equation}\label{eq:si_expanded}
s_i(t)
=
\sum_{k=0}^{K-1}\mathbb{P}\bigl(M_{-i}(t)=k\mid \mathcal{F}_{t-1}\bigr)\,
\mathbb{E}\!\left[\pi(D(t),k)\right],
\end{equation}
where the outer probability is w.r.t.\ the Poisson-binomial law \eqref{eq:poisson_binomial}
and the inner expectation is w.r.t.\ $D(t)$.
\end{definition}

\subsection{Belief Update Dynamics}
Let $A_i(t)\in\{0,1\}$ denote the allocation indicator for agent $i$ at epoch $t$.
Under the matching mechanism above,
\begin{equation}\label{eq:Ai_marginal}
A_i(t)\mid(\mathcal{F}_{t-1},B_i(t))
\sim
\begin{cases}
\mathrm{Bernoulli}\bigl(s_i(t)\bigr), & B_i(t)=1,\\
0, & B_i(t)=0.
\end{cases}
\end{equation}
We emphasize that $\{A_i(t) | i\in \mathcal{I}\}$ are generally \emph{not} independent across $i$, but \eqref{eq:Ai_marginal} correctly captures the marginal needed for belief updates.

Belief parameters are updated via the Beta--Bernoulli conjugate rule:
\begin{equation}\label{eq:beta_update}
\alpha_i(t+1)=\alpha_i(t)+A_i(t),\
\beta_i(t+1)=\beta_i(t)+(1-A_i(t))B_i(t).
\end{equation}
Define the total pseudo-count of  of successful and unsuccessful allocations as
\begin{equation}\label{eq:ni_def}
n_i(t):=\alpha_i(t)+\beta_i(t),
\end{equation}
which evolves as
\begin{equation}\label{eq:ni_update}
n_i(t+1)=n_i(t)+B_i(t).
\end{equation}
Combining \eqref{eq:belief_mean} and \eqref{eq:beta_update}--\eqref{eq:ni_update} and using \(A_i(t) = A_i(t) B_i(t)\) yields
\begin{equation}\label{eq:xi_update}
\begin{aligned}
x_i(t+1)&=\frac{\alpha_i(t)+A_i(t)}{n_i(t)+B_i(t)}\\
&=x_i(t)+\frac{B_i(t)}{n_i(t)+B_i(t)}\bigl(A_i(t)-x_i(t)\bigr),
\end{aligned}
\end{equation}
so beliefs are updated only upon participation.

\subsection{Mean-Field Approximation}
To study aggregate behavior, we restrict attention to stationary, uniform recommendation policies.

\begin{assumption}[Uniform recommendation intensity]\label{ass:uniform_u}
Recommendations are issued with a uniform, time-invariant intensity
$u_i(t)=u$ for all $i$ and $t$, with $u\in[0,1]$.
\end{assumption}

Define the empirical averages
\begin{equation}\label{eq:empirical_means}
\bar\alpha(t):=\frac{1}{K}\sum_{i=1}^K\alpha_i(t),\quad
\bar n(t):=\frac{1}{K}\sum_{i=1}^K n_i(t),\quad
\bar x(t):=\frac{\bar\alpha(t)}{\bar n(t)}.
\end{equation}
From \eqref{eq:beta_update} and \eqref{eq:ni_update},
\begin{equation}\label{eq:bar_updates_raw}
\bar\alpha(t+1)=\bar\alpha(t)+\frac{1}{K}\sum_{i=1}^K A_i(t)
\bar n(t+1)=\bar n(t)+\frac{1}{K}\sum_{i=1}^K B_i(t).
\end{equation}

Taking conditional expectation given $\mathcal{F}_{t-1}$ and using
$\mathbb{E}[A_i(t)\mid\mathcal{F}_{t-1}]=q_i(t)s_i(t)$ and
$\mathbb{E}[B_i(t)\mid\mathcal{F}_{t-1}]=q_i(t)$ yield
\begin{equation}\label{eq:bar_cond_exp}
\begin{aligned}
\mathbb{E}[\bar\alpha(t+1)\mid\mathcal{F}_{t-1}]
&=
\bar\alpha(t)+\frac{1}{K}\sum_{i=1}^K q_i(t)s_i(t)\\
\mathbb{E}[\bar n(t+1)\mid\mathcal{F}_{t-1}]
&=
\bar n(t)+\bar q(t),    
\end{aligned}
\end{equation}
where $\bar q(t):=\frac{1}{K}\sum_{i=1}^K q_i(t)$.

\begin{figure}[t]
\centering
\includegraphics[width=0.7\linewidth]{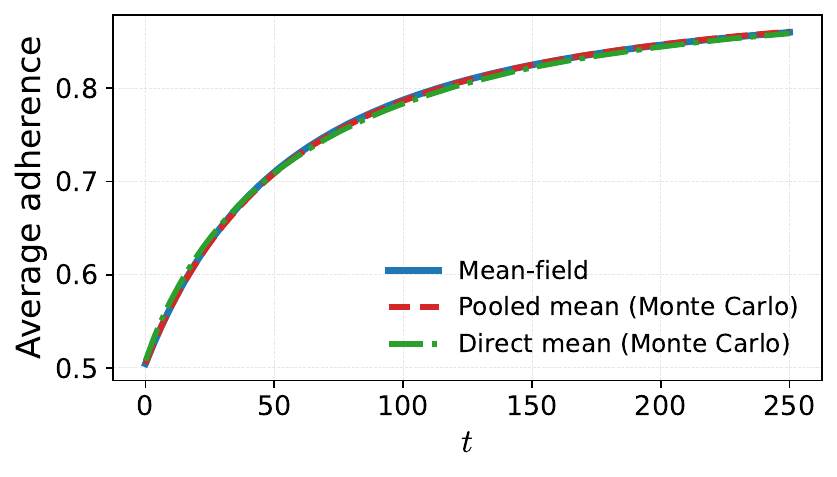}
\caption{Mean-field trajectory versus microscopic simulation for \(K=100\) agents under strong heterogeneity (\(\alpha_i(0),\beta_i(0)\sim U(1,50)\), \(p_i\sim U(0,1)\)). Demand follows \(D(t)\sim\mathrm{Poisson}(80)\) with recommendation intensity \(u=0.9\). The dashed and dotted curves show the Monte Carlo averages of the pooled and direct population means across \(100\) runs, illustrating that the mean-field approximation remains accurate even under large heterogeneity.}
\label{fig:mf_vs_micro}
\end{figure}

\begin{figure*}[t]
\centering
\begin{tikzpicture}[
    font=\small,
    >=Latex,
    node distance=11mm,
    box/.style={
        draw,
        rounded corners=3pt,
        align=center,
        minimum height=8mm,
        minimum width=20mm,
        inner sep=3pt,
        line width=0.6pt
    },
    micro/.style={box, draw=blue!60!black, fill=blue!6},
    mfbox/.style={box, draw=teal!60!black, fill=teal!6, minimum width=70mm},
    ctrl/.style={box, draw=purple!70!black, fill=purple!8, minimum width=16mm, minimum height=6mm},
    arr/.style={->, thick, draw=blue!60!black},
    arrctrl/.style={->, thick, draw=purple!70!black},
    exo/.style={box, draw=orange!70!black, fill=orange!10, minimum width=16mm, minimum height=6mm},
    exoarr/.style={->, thick, draw=orange!70!black},
    approx/.style={->, thick, dashed, draw=black!55},
    highlight/.style={draw=black!55, dashed, rounded corners=4pt, inner xsep=4pt, inner ysep=10pt}
]

% =========================
% Microscopic closed loop
% =========================
\node[micro] (belief) {\textbf{Belief}\\ $x_i(t)$};
\node[micro, right=of belief] (particip) {\textbf{Participation}\\ $q_i(t)$};
\node[micro, right=of particip] (congestion) {\textbf{Congestion}\\ $M_{-i}(t)$};
\node[micro, right=of congestion] (alloc) {\textbf{Allocation}\\ $s_i(t)$};
\node[micro, right=of alloc] (update) {\textbf{Update}\\ $x_i(t+1)$};

% forward arrows
\draw[arr] (belief) -- (particip);
\draw[arr] (particip) -- (congestion);
\draw[arr] (congestion) -- (alloc);
\draw[arr] (alloc) -- (update);

% feedback
\draw[arr]
(update.north) -- ++(0,8mm) -| (belief.north);

% exogenous demand
\node[exo, align=center, below=6mm of alloc] (demand) {\textbf{Demand}\\ $D(t)$};
\draw[exoarr] (demand.north) -- (alloc.south);

% control input
\node[ctrl, below=6mm of particip] (control) {\textbf{Control}\\ $u_i(t)$};
\draw[arrctrl] (control.north) -- (particip.south);

% dotted highlight around the microscopic block
\node[highlight, fit=(belief)(particip)(congestion)(alloc)(update)(control)(demand)] (microblock) {};

% label for microscopic model
\node[above=4mm of microblock, font=\bfseries] 
{Microscopic Model of Adherence-Coupled Matching Dynamics};

% =========================
% Mean-field block
% =========================
\node[mfbox, below=26mm of congestion] (mf) {
\textbf{Mean-Field Approximation}\\[2pt]
$\displaystyle
\bar x(t+1) =
\bar x(t) + \frac{\bar q(t)}{\bar n(t)+\bar q(t)} \bigl(s(t)-\bar x(t)\bigr)$\\

$\displaystyle
s(t) = \E\!\left[\min\!\left(1,\frac{D(t)}{1+(K-1)\bar q(t)}\right)\right], \quad
\bar n(t+1) = \bar n(t) + \bar q(t), \quad
\bar q(t) = (1-\bar x(t))\,p+u\,\bar x(t)$};

% arrow from microscopic to mean-field
\draw[approx]
(microblock.south)
-- node[right, xshift=1mm]{\scriptsize large-population closure}
(mf.north);

\end{tikzpicture}
\caption{Closed-loop adherence dynamics: microscopic stochastic model (finite $K$) and its mean-field approximation obtained via large-population averaging.}
\end{figure*}
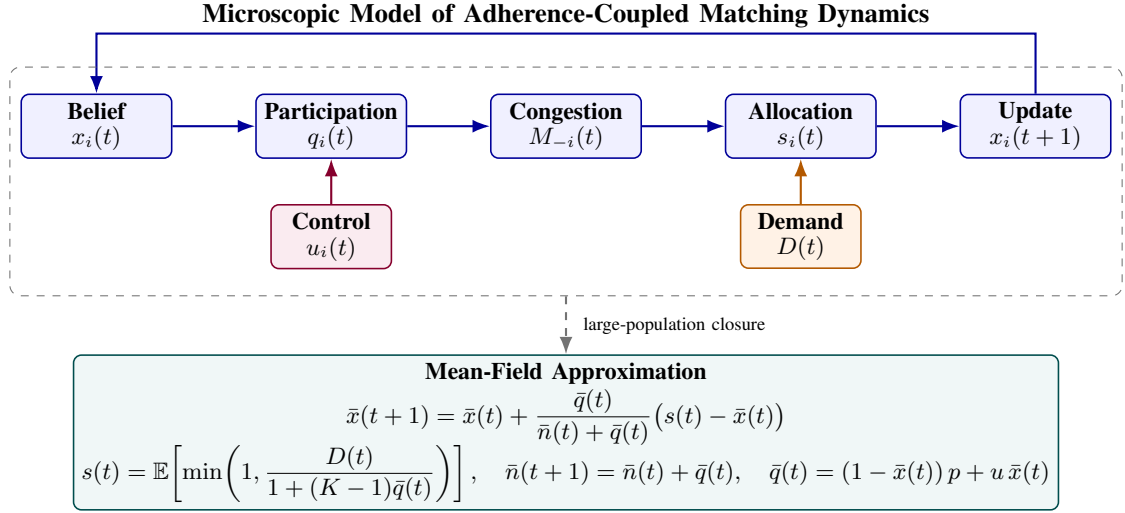

\paragraph{Deterministic closure}
Motivated by large-population averaging (law of large numbers and concentration),
we adopt the mean-field approximation
\begin{equation}\label{eq:mf_closure}
\frac{1}{K}\sum_{i=1}^K q_i(t)s_i(t)\approx \bar q(t)\,s(t),
\quad
\text{where} \ s_i(t)\approx s(t) \ \forall \ i,
\end{equation} 
where $s(t)$ is the representative allocation probability under average congestion.
Under Assumption~\ref{ass:uniform_u}, we further approximate
\begin{equation}\label{eq:qbar_mf}
\bar q(t)\approx (1-\bar x(t))\,p + u\,\bar x(t),
\qquad
p:=\frac{1}{K}\sum_{i=1}^K p_i.
\end{equation}
Finally, we approximate the congestion faced by a tagged agent by its mean
$\mathbb{E}[M_{-i}(t)]\approx (K-1)\bar q(t)$, which yields
\begin{equation}\label{eq:s_mf}
s(t)
=
\mathbb{E}\!\left[\min\!\left(1,\frac{D(t)}{1+(K-1)\bar q(t)}\right)\right].
\end{equation}

Substituting \eqref{eq:mf_closure}--\eqref{eq:s_mf} into \eqref{eq:bar_updates_raw} produces the deterministic mean-field recursion
\begin{equation}\label{eq:mf_alpha_n}
\bar\alpha(t+1)=\bar\alpha(t)+\bar q(t)s(t),\qquad
\bar n(t+1)=\bar n(t)+\bar q(t).
\end{equation}
Since $\bar x(t)=\bar\alpha(t)/\bar n(t)$, we obtain the closed recursion
\begin{equation}\label{eq:mf_x}
\bar x(t+1)
=
\bar x(t)
+
\frac{\bar q(t)}{\bar n(t)+\bar q(t)}
\bigl(s(t)-\bar x(t)\bigr),
\end{equation}
with $\bar q(t)$ given by \eqref{eq:qbar_mf} and $s(t)$ by \eqref{eq:s_mf}.

\begin{exmp}[Accuracy of the mean-field approximation]\label{ex:mf_accuracy}
For run \(m\), let \( \bar{x}_K^{(m)}(t)=\frac{1}{K}\sum_{i=1}^{K}x_i^{(m)}(t) \) and \( \bar{x}_{\mathrm{pool}}^{(m)}(t)=\frac{\sum_i \alpha_i^{(m)}(t)}{\sum_i n_i^{(m)}(t)} \). Figure~\ref{fig:mf_vs_micro} shows the mean-field trajectory and the Monte Carlo averages over \(m=100\) runs, demonstrating close agreement.
\end{exmp}

\subsection{Control-Dependent Performance Metrics}
Fix a uniform time-invariant control policy $u(t)\equiv u\in[0,1]$ and let $\bar x(t;u)$ denote the mean-field trajectory \eqref{eq:mf_x}.
Define the induced participation and allocation probabilities
\begin{equation*}
\begin{aligned}
\bar q(t;u)&:=(1-\bar x(t;u))p+u\,\bar x(t;u),\\
s(t;u)&:=
\mathbb{E}\!\left[\min\!\left(1,\frac{D(t)}{1+(K-1)\bar q(t;u)}\right)\right].
\end{aligned}
\end{equation*}

Define steady-state adherence and throughput
\begin{equation}\label{eq:steadystate_defs}
x^\infty(u):=\lim_{t\to\infty}\bar x(t;u),
\qquad
\Gamma^\infty(u):=\lim_{t\to\infty}\bar q(t;u)\,s(t;u).
\end{equation}
We also define the $\varepsilon$-convergence time
\[
T_\varepsilon(u)
:=
\inf\Bigl\{t\ge0:\ |\bar x(\tau;u)-x^\infty(u)|\le\varepsilon\ \ \forall\tau\ge t\Bigr\}.
\]

\section{Theoretical Analysis} \label{sec:theory}
We analyze the deterministic mean-field recursion under a uniform time-invariant control $u$
and time-homogeneous demand.

\subsection{Standing assumptions for analysis}
Throughout this section, assume $D(t)\overset{\text{i.i.d.}}{\sim}\mathrm{Poisson}(\lambda)$ with $\lambda>0$ and $u(t)\equiv u$.
Define the auxiliary function, for $a>0$,
\begin{equation}\label{eq:g_def}
g(a):=\mathbb{E}\!\left[\min\!\left(1,\frac{D}{a}\right)\right],\qquad D\sim\mathrm{Poisson}(\lambda).
\end{equation}
Then the mean-field allocation term can be written as $s(t)=g(1+(K-1)\bar q(t))$.

% ============================================================
\subsection{Well-Posedness of the Mean-Field Recursion}
% ============================================================
\begin{lemma}[Well-posedness and invariance]\label{lem:wellposed}
Fix $K\in\mathbb{N}$ and $p,u\in[0,1]$. Let $\bar x(0)\in[0,1]$ and $\bar n(0)>0$.
For $t\ge0$ define
\begin{align}
\bar q(t)&:= (1-\bar x(t))p+u\bar x(t),\label{eq:mf_q_def_re}\\
s(t)&:= g\bigl(1+(K-1)\bar q(t)\bigr),\label{eq:mf_s_def_re}\\
\bar n(t+1)&:=\bar n(t)+\bar q(t),\label{eq:mf_n_update_re}\\
\bar x(t+1)&:=\bar x(t)+\frac{\bar q(t)}{\bar n(t)+\bar q(t)}\bigl(s(t)-\bar x(t)\bigr).\label{eq:mf_x_update_re}
\end{align}
Then $(\bar x(t),\bar n(t))$ is uniquely defined for all $t\ge0$ and satisfies
\[
\bar x(t)\in[0,1],\qquad \bar n(t)\ge \bar n(t-1)>0.
\]
\end{lemma}

\begin{proof}
The proof is by induction on $t$.
The base case holds by assumption.
If $\bar x(t)\in[0,1]$, then $\bar q(t)\in[0,1]$ by \eqref{eq:mf_q_def_re}, and since $g(\cdot)\in[0,1]$,
we have $s(t)\in[0,1]$. Equation \eqref{eq:mf_n_update_re} gives $\bar n(t+1)\ge \bar n(t)>0$.
Let $\gamma(t):=\bar q(t)/(\bar n(t)+\bar q(t))\in[0,1]$, then \eqref{eq:mf_x_update_re} becomes
$\bar x(t+1)=(1-\gamma(t))\bar x(t)+\gamma(t)s(t)$, a convex combination of $[0,1]$ values.
Thus $\bar x(t+1)\in[0,1]$, completing the induction.
Uniqueness follows since \eqref{eq:mf_n_update_re}--\eqref{eq:mf_x_update_re} define a deterministic map.
\end{proof}

% ============================================================
\subsection{Equilibrium Analysis}
% ============================================================
We study fixed points under i.i.d.\ $D\sim\mathrm{Poisson}(\lambda)$.
Define, for $x\in[0,1]$,
\begin{equation}\label{eq:q_a_s_of_x}
\begin{aligned}
&q(x):=p+(u-p)x, \ a(x):=1+(K-1)q(x),\\
&s(x):=g(a(x)).    
\end{aligned}
\end{equation}

\begin{theorem}[Existence]\label{thm:eq_exist}
Fix $K\in\mathbb N$, $p,u\in[0,1]$, and $\lambda>0$. Then there exists $x^\star\in[0,1]$ such that
\begin{equation}\label{eq:fixed_point}
x^\star=s(x^\star)=g\bigl(1+(K-1)(p+(u-p)x^\star)\bigr).
\end{equation}
\end{theorem}

\begin{proof}
The map $g$ is continuous and $g(a)\in[0,1]$ for all $a>0$ by dominated convergence.
Hence $s(\cdot)$ is continuous and maps $[0,1]$ into $[0,1]$.
Let $f(x):=s(x)-x$. Then $f(0)=s(0)\ge0$ and $f(1)=s(1)-1\le0$.
By the intermediate value theorem, $f$ has a root in $[0,1]$.
\end{proof}

\begin{lemma}[Slope of $g$ (Poisson demand)]\label{lem:g_derivative}
Let $D\sim\mathrm{Poisson}(\lambda)$ and $g$ be defined in \eqref{eq:g_def}.
Then $g$ is continuous and nonincreasing on $(0,\infty)$, and for all $a\notin\mathbb{Z}$,
\begin{equation}\label{eq:gprime}
g'(a)=-\frac{\lambda}{a^2}\,\mathbb{P}\bigl(D\le \lfloor a\rfloor-1\bigr).
\end{equation}
Moreover, $g$ is globally Lipschitz on any compact interval $[a_{\min},a_{\max}]\subset(0,\infty)$.
\end{lemma}

\begin{proof}
Monotonicity and continuity follow since $a\mapsto \min(1,D/a)$ is continuous and nonincreasing for each fixed $D$,
and is dominated by $1$.
For $a\in(n,n+1)$ with $n\in\mathbb{Z}_{\ge0}$, we have $\min(1,D/a)=D/a$ on $\{D\le n\}$ and $=1$ on $\{D\ge n+1\}$,
so $g'(a)=-(1/a^2)\mathbb{E}[D\mathbf{1}_{\{D\le n\}}]$.
Using $\mathbb{E}[D\mathbf{1}_{\{D\le n\}}]=\lambda\,\mathbb{P}(D\le n-1)$ for Poisson $D$
yields \eqref{eq:gprime}. Boundedness of $g'$ on each smooth sub-interval implies Lipschitzness on compacts.
\end{proof}

\begin{theorem}[Uniqueness]\label{thm:eq_unique}
Adopt the setting of \Cref{thm:eq_exist}. Define
\[
a_{\min}:=1+(K-1)u,\qquad a_{\max}:=1+(K-1)p.
\]
\begin{enumerate}
\item If $u\ge p$, then the fixed point \eqref{eq:fixed_point} is unique on $[0,1]$.
\item If $u<p$ and
\begin{equation*}\label{eq:exact_unique_condition_re}
\begin{aligned}
&(K-1)(p-u)\,\lambda\,
\max_{a\in\mathcal{A}}
\frac{\mathbb{P}(D\le \lfloor a\rfloor-1)}{a^2}
<1,\\
&\mathcal{A}:=\{a_{\min}\}\cup\bigl(\mathbb{Z}\cap(a_{\min},a_{\max}]\bigr),
\end{aligned}
\end{equation*}
then the fixed point is unique.
\end{enumerate}
\end{theorem}

\begin{proof}
Write $s(x)=g(a(x))$ with $a(x)=1+(K-1)(p+(u-p)x)$.
If $u\ge p$, then $a(x)$ is nondecreasing in $x$ and $g$ is nonincreasing, hence $s$ is nonincreasing.
Therefore $f(x):=s(x)-x$ is continuous and strictly decreasing on $[0,1]$, so it has at most one root;
existence follows from \Cref{thm:eq_exist}.

If $u<p$, then $a(x)$ is strictly decreasing in $x$ and maps $[0,1]$ onto $[a_{\min},a_{\max}]$.
By \Cref{lem:g_derivative}, $g$ is continuous and piecewise $C^1$ with slope \eqref{eq:gprime}.
On each interval $(n,n+1)$ the magnitude $|g'(a)|=\lambda\,\mathbb{P}(D\le n-1)/a^2$ is decreasing in $a$,
so the maximum slope over $[a_{\min},a_{\max}]$ is attained at the left endpoint of each smooth piece,
captured by the set $\mathcal{A}$. Hence $g$ is Lipschitz on $[a_{\min},a_{\max}]$ with constant
\[
L_g:=\lambda\max_{a\in\mathcal{A}}\frac{\mathbb{P}(D\le \lfloor a\rfloor-1)}{a^2}.
\]
Since $|a'(x)|=(K-1)(p-u)$, the composition $s(x)=g(a(x))$ is Lipschitz on $[0,1]$ with constant
$L=(K-1)(p-u)L_g$.
Under \eqref{eq:exact_unique_condition_re}, $L<1$, so $s$ is a contraction and admits a unique fixed point.
\end{proof}

\begin{figure}[tb]
\centering
\includegraphics[width=0.94\linewidth]{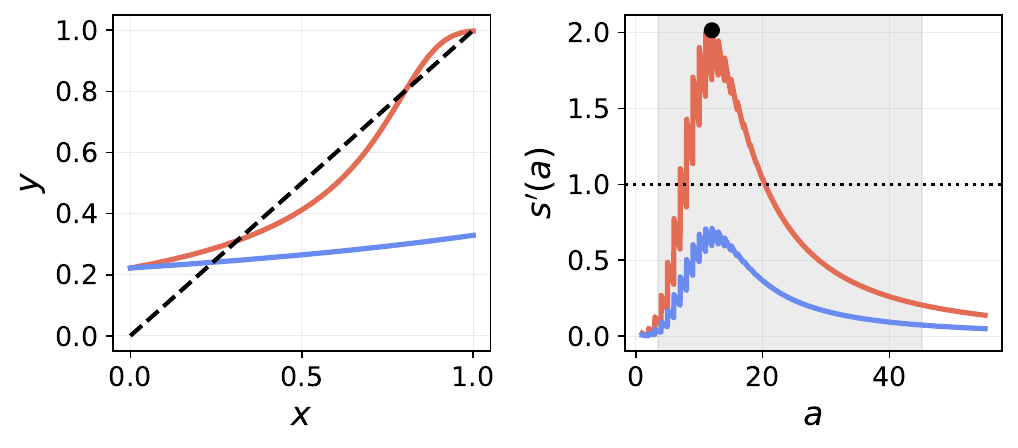}
\caption{Illustration of equilibrium structure. 
Left: the fixed-point map $y=s(x)$ and the diagonal $y=x$ (see \eqref{eq:q_a_s_of_x}). 
For $u=0.6$ the curves intersect once, yielding a unique equilibrium, while for $u=0.05$ multiple intersections appear. 
Right: derivative $s'(a)$ over the interval $[a_{\min},a_{\max}]$. 
When $\sup s'(a)<1$ the contraction condition of \Cref{thm:eq_unique} holds (unique equilibrium), whereas when $\sup s'(a)>1$ the condition fails, allowing multiple equilibria.}
\label{fig:equilibrium_derivative}
\end{figure}

\begin{exmp}[Non-unique equilibria]\label{exmp:nonunique}
Consider $K=50$, $p=0.9$, and $\lambda=10$.  
For $u=0.05$, the uniqueness condition in \Cref{thm:eq_unique} is violated.  
As shown in \Cref{fig:equilibrium_derivative}, the curves $y=s(x)$ and $y=x$ intersect multiple times, implying multiple equilibria.  
In contrast, for $u=0.6$ the condition holds and the curves intersect once, yielding a unique equilibrium.
\end{exmp}

\subsection{Global Convergence for $u\ge p$}

\begin{theorem}[Global convergence for $u\ge p$]\label{thm:global_convergence}
Let $p>0$, $u\in[p,1]$, and $D(t)\overset{\text{i.i.d.}}{\sim}\mathrm{Poisson}(\lambda)$ with $\lambda>0$.
Let $(\bar x(t),\bar n(t))$ evolve according to \Cref{lem:wellposed}, and define
$s(x):=g(1+(K-1)(p+(u-p)x))$.
Then the fixed point $x^\star\in[0,1]$ uniquely exists, satisfying $x^\star=s(x^\star)$, and for every
$\bar x(0)\in[0,1]$, $\bar n(0)>0$,
\[
\bar x(t)\to x^\star,\qquad \bar n(t)\to\infty,\qquad
\gamma(t):=\frac{\bar q(t)}{\bar n(t)+\bar q(t)}\to0.
\]
\end{theorem}

\begin{proof}
By \Cref{lem:wellposed}, $\bar x(t)\in[0,1]$ and $\bar n(t)>0$ for all $t$.

For $u\ge p$, the map $q(x)=p+(u-p)x$ is nondecreasing, hence
$a(x):=1+(K-1)q(x)$ is nondecreasing. Since
$s(x)=\mathbb E\!\left[\min\!\left(1,\frac{D}{a(x)}\right)\right]$
is nonincreasing in $a(x)$, it follows that $s$ is nonincreasing in $x$.
Therefore $f(x):=s(x)-x$ is continuous and strictly decreasing on $[0,1]$.
Because $s([0,1])\subset[0,1]$, we have $f(0)=s(0)\ge0$ and $f(1)=s(1)-1\le0$,
so there exists a unique $x^\star\in[0,1]$ such that $f(x^\star)=0$.

Since $\bar x(t)\in[0,1]$, we have $p\le \bar q(t)=p+(u-p)\bar x(t)\le1$.
Hence $\bar n(t+1)=\bar n(t)+\bar q(t)$ implies
$\bar n(0)+tp\le \bar n(t)\le \bar n(0)+t$, so $\bar n(t)\to\infty$.
Therefore $\gamma(t)=\bar q(t)/(\bar n(t)+\bar q(t))\to0$.
Moreover, $\gamma(t)\le1/\bar n(t)\le1/(\bar n(0)+tp)$, so
$\gamma(t)^2\le1/(\bar n(0)+tp)^2$ and therefore $\sum_t\gamma(t)^2<\infty$.
Also $\gamma(t)\ge p/(\bar n(t)+1)\ge p/(\bar n(0)+t+1)$, hence
$\sum_t\gamma(t)=\infty$.

Define $V(t):=(\bar x(t)-x^\star)^2$. Using
$\bar x(t+1)=\bar x(t)+\gamma(t)f(\bar x(t))$, we obtain
\begin{equation}
V(t+1)=V(t)+2\gamma(t)(\bar x(t)-x^\star)f(\bar x(t))+\gamma(t)^2f(\bar x(t))^2 .
\end{equation}
Since $f$ is nonincreasing and $f(x^\star)=0$, we have
$(\bar x(t)-x^\star)f(\bar x(t))\le0$.
Because $s(x),x\in[0,1]$, we also have $|f(x)|\le1$, so
$V(t+1)\le V(t)+\gamma(t)^2$.

Define $W(t):=V(t)+\sum_{k=t}^{\infty}\gamma(k)^2$. Then
$W(t+1)\le W(t)$, so $W(t)$ converges. Since
$\sum_{k=t}^{\infty}\gamma(k)^2\to0$, it follows that $V(t)$ converges
to some $L\ge0$.

Suppose $L>0$. Then for all sufficiently large $t$,
$|\bar x(t)-x^\star|\ge\sqrt L/2$.
Because $f$ is continuous, strictly decreasing, and $f(x^\star)=0$,
there exists $c>0$ such that
$(x-x^\star)f(x)\le -c$ whenever $|x-x^\star|\ge\sqrt L/2$.
Hence for all large $t$,
$V(t+1)\le V(t)-2c\,\gamma(t)+\gamma(t)^2$.
Since $\gamma(t)\to0$, we have $\gamma(t)^2\le c\,\gamma(t)$ for large $t$,
so $V(t+1)\le V(t)-c\,\gamma(t)$.

Summing over $t$ gives a contradiction because $\sum_t\gamma(t)=\infty$
while $V(t)\ge0$. Therefore $L=0$ and $\bar x(t)\to x^\star$.
The properties $\bar n(t)\to\infty$ and $\gamma(t)\to0$ were already shown.
\end{proof}

% \begin{figure}[t]
% \centering
% \includegraphics[width=0.7\linewidth]{mf_vs_micro.pdf}
% \caption{Comparison of the deterministic mean-field trajectory with the microscopic simulation of $K=100$ agents. 
% The microscopic system follows the Beta--Bernoulli update with parameters $p=0.6$, $u=0.9$, demand rate $\lambda=80$, and initial beliefs $(\alpha_0,\beta_0)=(2,2)$. 
% The dashed curve shows the Monte-Carlo average of the microscopic population mean across $40$ simulation runs, while the shaded region represents the $10$--$90$ percentile variability. 
% The close agreement demonstrates that the mean-field recursion accurately captures the aggregate dynamics of the stochastic system.}
% \label{fig:mf_vs_micro}
% \end{figure}

% \begin{exmp}[Accuracy of the mean-field approximation]\label{ex:mf_accuracy}
% We compare the deterministic mean-field recursion \eqref{eq:mf_n_update_re}--\eqref{eq:mf_x_update_re} with the microscopic stochastic dynamics of $K$ agents following the Beta--Bernoulli learning rule with demand $D(t)\sim\mathrm{Poisson}(\lambda)$. 
% The population average $\bar x_K(t)=\frac{1}{K}\sum_{i=1}^K x_i(t)$ from Monte-Carlo simulations is compared with the mean-field trajectory $\bar x(t)$. 
% As shown in \Cref{fig:mf_vs_micro}, the mean-field closely matches the microscopic average with only small finite-population fluctuations.
% \end{exmp}

% ============================================================
\subsection{Control-Dependent Performance Metrics}
% ============================================================
The steady-state adherence $x^\infty(u)$ and throughput $\Gamma^\infty(u)$ are defined in
\eqref{eq:steadystate_defs}.

\begin{lemma}[Existence of steady-state metrics for $u\ge p$]\label{lem:metrics_exist_up}
Assume $u\in[p,1]$ and $D(t)\overset{\text{i.i.d.}}{\sim}\mathrm{Poisson}(\lambda)$.
Then $x^\infty(u)$ and $\Gamma^\infty(u)$ exist and satisfy
\[
x^\infty(u)=x^\star(u),\qquad
\Gamma^\infty(u)=q^\star(u)\,g\!\bigl(1+(K-1)q^\star(u)\bigr),
\]
where $x^\star(u)$ is the unique solution to \eqref{eq:fixed_point} and
$q^\star(u):=p+(u-p)x^\star(u)$.
\end{lemma}

\begin{proof}
By \Cref{thm:global_convergence}, $\bar x(t;u)\to x^\star(u)$.
Thus $\bar q(t;u)\to q^\star(u)$ and continuity of $g$ gives
$s(t;u)=g(1+(K-1)\bar q(t;u))\to g(1+(K-1)q^\star(u))$.
\end{proof}

\begin{lemma}[Monotone steady-state adherence for $u\ge p$]\label{lem:x_monotone}
Assume $u\in[p,1]$ and $D(t)\overset{\text{i.i.d.}}{\sim}\mathrm{Poisson}(\lambda)$.
Then $u\mapsto x^\infty(u)$ is nonincreasing on $[p,1]$.
\end{lemma}

\begin{proof}
For fixed $x$, the map $u\mapsto q(x)=p+(u-p)x$ is nondecreasing.
Since $g$ is nonincreasing, $u\mapsto s(x,u):=g(1+(K-1)q(x))$ is nonincreasing.
For each $u\ge p$, $x\mapsto s(x,u)$ is nonincreasing, hence $F(x,u):=s(x,u)-x$ is strictly decreasing in $x$
and has a unique zero $x^\star(u)$.
If $u_1\le u_2$, then $F(\cdot,u_2)\le F(\cdot,u_1)$ pointwise, so the unique root satisfies
$x^\star(u_2)\le x^\star(u_1)$.
\end{proof}

\begin{theorem}[Local throughput gain implies distinct optima]\label{thm:local_gain_no_coincide}
Assume $u\in[p,1]$ and $D\sim\mathrm{Poisson}(\lambda)$.
Let $a_p:=1+(K-1)p$ and $x_p:=g(a_p)$. Let $x^\star(u)$ solve \eqref{eq:fixed_point}.
Then, at equilibrium,
\[
\Gamma^\infty(u)=q^\star(u)\,x^\star(u)=\bigl(p+(u-p)x^\star(u)\bigr)x^\star(u).
\]
If $g$ is differentiable at $a_p$ (e.g., $a_p\notin\mathbb{Z}$), then $\Gamma^\infty$ is differentiable at $u=p$ and
\begin{equation}\label{eq:Gamma_prime_p_thm_re}
\Gamma^{\infty\,\prime}(p)=x_p^2+p(K-1)x_p\,g'(a_p).
\end{equation}
If $\Gamma^{\infty\,\prime}(p)>0$, then $\exists\,\varepsilon>0$ such that
$\Gamma^\infty(u)>\Gamma^\infty(p)$ for all $u\in(p,p+\varepsilon)$; combined with \Cref{lem:x_monotone},
no constant control simultaneously maximizes $x^\infty(u)$ and $\Gamma^\infty(u)$ on $[p,1]$.
\end{theorem}

\begin{proof}
At $u=p$, $x^\star(p)=x_p$.
Differentiating $x=g(1+(K-1)(p+(u-p)x))$ at $u=p$ gives
$x^{\star\prime}(p)=(K-1)x_p g'(a_p)$.
Differentiating $\Gamma^\infty(u)=p x^\star(u)+(u-p)(x^\star(u))^2$ at $u=p$ yields \eqref{eq:Gamma_prime_p_thm_re}.
\end{proof}

% ============================================================
\subsection{Geometry of the Steady-State Performance Frontier}
% ============================================================
\begin{theorem}[Frontier geometry]\label{thm:frontier_geometry}
Assume $D(t)\overset{\text{i.i.d.}}{\sim}\mathrm{Poisson}(\lambda)$, $p\in(0,1]$, and $u\in[p,1]$.
Let $x^\infty(u)=x^\star(u)$ denote the unique solution to \eqref{eq:fixed_point}, and define
$q^\star(u):=p+(u-p)x^\star(u)$, $a^\star(u):=1+(K-1)q^\star(u)$, and
$\Gamma^\infty(u):=q^\star(u)g(a^\star(u))$. Then:
\begin{enumerate}
\item $u\mapsto x^\infty(u)$ and $u\mapsto \Gamma^\infty(u)$ are continuous on $[p,1]$.

\item $u\mapsto x^\infty(u)$ is nonincreasing on $[p,1]$.

\item If $a^\star(u)\notin\mathbb{Z}$ for all $u\in[p,1]$ and
\begin{equation}\label{eq:throughput_mono_condition_re}
\begin{aligned}
&(K-1)\sup_{u\in[p,1]}|g'(a^\star(u))|
\sup_{u\in[p,1]}\bigl(p+2(u-p)x^\infty(u)\bigr)\\
&<
\inf_{u\in[p,1]} x^\infty(u),
\end{aligned}
\end{equation}
then $u\mapsto \Gamma^\infty(u)$ is strictly increasing on $[p,1]$, so the frontier is strictly decreasing
in the $(x,\Gamma)$-plane.
\end{enumerate}
\end{theorem}

\begin{proof}
Since $g$ is continuous and nonincreasing, the map 
$s(x,u)=g(1+(K-1)(p+(u-p)x))$ is continuous in $(x,u)$. 
Because the equilibrium $x^\infty(u)$ solving $x=s(x,u)$ is unique for $u\ge p$, 
it depends continuously on $u$. Since 
$\Gamma^\infty(u)=p x^\infty(u)+(u-p)(x^\infty(u))^2$, 
it follows that both $u\mapsto x^\infty(u)$ and $u\mapsto \Gamma^\infty(u)$ are continuous on $[p,1]$. 
Monotonicity of $x^\infty(u)$ follows from Lemma~\ref{lem:x_monotone}.

For monotone throughput, note $\Gamma^\infty(u)=q^\star(u)x^\infty(u)=p x^\infty(u)+(u-p)(x^\infty(u))^2$.
At differentiability points,
\[
\Gamma^{\infty\,\prime}(u)=(x^\infty(u))^2+x^{\infty\,\prime}(u)\bigl(p+2(u-p)x^\infty(u)\bigr).
\]
Since $x^{\infty\,\prime}(u)\le 0$, we bound the negative term using 
$|x^{\infty\,\prime}(u)|\le (K-1)x^\infty(u)|g'(a^\star(u))|$, 
a conservative estimate obtained from implicit differentiation of \eqref{eq:fixed_point}. 
Substituting yields 
$\Gamma^{\infty\,\prime}(u)\ge x^\infty(u)\!\left[x^\infty(u)-(K-1)|g'(a^\star(u))|\bigl(p+2(u-p)x^\infty(u)\bigr)\right]$, 
which is positive under \eqref{eq:throughput_mono_condition_re}.
\end{proof}

% ============================================================
\subsection{Optimal Constant Control Under an Adherence Constraint}
% ============================================================
We restrict to constant controls $u\in[p,1]$, under which the mean-field steady state is (unique and) well-defined by the fixed point $x^\infty(u)=x^\star(u)$ (cf. \Cref{lem:metrics_exist_up}).

\begin{problem}[Throughput maximization with adherence floor]\label{prob:const_opt}
Fix $\underline x\in(0,1)$ and choose $u\in[p,1]$ to
\[
\max\ \Gamma^\infty(u)
\qquad \text{s.t.}\qquad x^\infty(u)\ge \underline x.
\]
\end{problem}

\begin{theorem}[Solution via the frontier]\label{thm:const_opt_solution}
Assume \Cref{lem:x_monotone},\Cref{thm:frontier_geometry} and \eqref{eq:throughput_mono_condition_re},
so that $x^\infty(\cdot)$ is continuous and nonincreasing and $\Gamma^\infty(\cdot)$ is continuous and strictly increasing on $[p,1]$.
Define the feasible set $\mathcal{U}(\underline x):=\{u\in[p,1]: x^\infty(u)\ge \underline x\}$ and
\[
u_{\max}(\underline x):=\sup \mathcal{U}(\underline x).
\]
Then:
\begin{enumerate}
\item If $\mathcal{U}(\underline x)=\emptyset$ (equivalently $\underline x> x^\infty(p)$), then \Cref{prob:const_opt} is infeasible.
\item If $\mathcal{U}(\underline x)\neq\emptyset$, then $\mathcal{U}(\underline x)$ is a closed interval of the form $[p,u_{\max}(\underline x)]$ and every optimizer of \Cref{prob:const_opt} equals
\[
u^\star(\underline x)=u_{\max}(\underline x).
\]
\item (Constraint activity.) If $u^\star(\underline x)<1$ (equivalently $\underline x> x^\infty(1)$), then the adherence constraint is active at the optimum:
\[
x^\infty\bigl(u^\star(\underline x)\bigr)=\underline x.
\]
If instead $\underline x\le x^\infty(1)$, then $u^\star(\underline x)=1$ and the constraint may be slack (it is active only if $\underline x=x^\infty(1)$).
\end{enumerate}
\end{theorem}

\begin{proof}
By monotonicity of $x^\infty$, the feasible set $\mathcal{U}(\underline x)$ is either empty or an interval of the form $[p,\sup\mathcal{U}(\underline x)]$.
By continuity of $x^\infty$, $\mathcal{U}(\underline x)$ is closed, hence the supremum is attained and
$\mathcal{U}(\underline x)=[p,u_{\max}(\underline x)]$.

Since $\Gamma^\infty$ is strictly increasing, the optimal value over a nonempty interval is attained uniquely at its right endpoint $u_{\max}(\underline x)$.

For constraint activity, suppose $u^\star(\underline x)<1$ and $x^\infty(u^\star(\underline x))>\underline x$.
By continuity, there exists $\eta>0$ such that $x^\infty(u)>\underline x$ for all $u\in[u^\star(\underline x),u^\star(\underline x)+\eta]$,
contradicting the definition of $u^\star(\underline x)$ as the maximal feasible control. Hence $x^\infty(u^\star(\underline x))=\underline x$.
If $\underline x\le x^\infty(1)$ then $u=1$ is feasible, so $u^\star(\underline x)=1$.
\end{proof}

% ============================================================
\subsection{Fast Algorithm for Optimal Constant Control}
% ============================================================
We compute $u^\star(\underline x)$ by bisection, exploiting that $u\mapsto x^\infty(u)$ is nonincreasing on $[p,1]$
(\Cref{lem:x_monotone}). Each feasibility check is performed by solving the equilibrium condition
$x^\infty(u)=x^\star(u)$ (instead of simulating the transient recursion).

\paragraph{Closed-form evaluation of $g(a)$.}
Let $D\sim\mathrm{Poisson}(\lambda)$ and define
\[
g(a):=\mathbb{E}\!\left[\min\!\left(1,\frac{D}{a}\right)\right],\qquad a>0.
\]
Let $k_0:=\lceil a\rceil$ and $F(k):=\mathbb{P}(D\le k)$, with the convention $F(k)=0$ for $k<0$.
Then
\begin{equation}\label{eq:g_fast}
\begin{aligned}
g(a)
&=\mathbb{P}(D\ge k_0)+\frac{1}{a}\mathbb{E}\!\left[D\,\mathbf{1}_{\{D\le k_0-1\}}\right]\\
&=1-F(k_0-1)+\frac{\lambda}{a}F(k_0-2),
\end{aligned}
\end{equation}
using $\mathbb{E}[D\,\mathbf{1}_{\{D\le m\}}]=\lambda\,\mathbb{P}(D\le m-1)$ for Poisson $D$.
Since in our model $a=1+(K-1)\bar q\in[1,K]$, pre-computing $F(k)$ for $k=0,1,\dots,K$ makes the time complexity of each $g(a)$ evaluation $O(1)$.

\paragraph{Equilibrium solver.}
For $u\in[p,1]$, the steady state $x^\infty(u)=x^\star(u)$ is the unique root of
\[
\Phi(x;u):=g\!\bigl(1+(K-1)(p+(u-p)x)\bigr)-x=0,\qquad x\in[0,1],
\]
since $x\mapsto \Phi(x;u)$ is continuous and strictly decreasing on $[0,1]$ (cf.\ \Cref{thm:eq_unique} and the monotonicity arguments in \Cref{thm:global_convergence}).
Hence $x^\star(u)$ can be computed by bisection in $x$.

\begin{algorithm}[t]
\caption{Bisection for $u^\star(\underline x)$ via equilibrium solves}\label{alg:fast_u_star}
\begin{algorithmic}[1]
\REQUIRE $K,p,\lambda$, adherence floor $\underline x\in(0,1)$, tolerances $\delta_u,\delta_x$.
\STATE Precompute $F(k)=\mathbb{P}(D\le k)$ for $k=0,1,\dots,K$; evaluate $g(\cdot)$ via \eqref{eq:g_fast}.
\STATE \textbf{Function} $\textsc{XStar}(u)$: bisection on $x\in[0,1]$ to find the unique root of
$\Phi(x;u)=g(1+(K-1)(p+(u-p)x))-x$ up to tolerance $\delta_x$; return $x^\star(u)$.
\STATE \textbf{Function} $\textsc{Feasible}(u)$: return TRUE iff $\textsc{XStar}(u)\ge \underline x$.
\STATE Initialize $\ell\leftarrow p$, $h\leftarrow 1$.
\IF{$\textsc{Feasible}(\ell)$ is FALSE} \STATE \textbf{return} INFEASIBLE. \ENDIF
\IF{$\textsc{Feasible}(h)$ is TRUE} \STATE \textbf{return} $u^\star(\underline x)\leftarrow 1$. \ENDIF
\WHILE{$h-\ell>\delta_u$}
\STATE $m\leftarrow(\ell+h)/2$.
\IF{$\textsc{Feasible}(m)$}
\STATE $\ell\leftarrow m$ \COMMENT{push right: seek maximal feasible $u$}
\ELSE
\STATE $h\leftarrow m$
\ENDIF
\ENDWHILE
\STATE \textbf{return} $u^\star(\underline x)\leftarrow \ell$.
\end{algorithmic}
\end{algorithm}

\paragraph{Throughput at the optimum.}
Once $x^\star(u)$ is computed, the steady-state throughput follows from
\[
\Gamma^\infty(u)=q^\star(u)\,x^\star(u)=\bigl(p+(u-p)x^\star(u)\bigr)x^\star(u),
\]
since at equilibrium $x^\star(u)=s(x^\star(u))$.

\section{Numerical Simulation}\label{sec:numerical}

We illustrate the theoretical results using the mean-field recursion 
\eqref{eq:mf_n_update_re}--\eqref{eq:mf_x_update_re} with parameters
$K=100$, $p=0.3$, $\lambda=50$, $\bar x(0)=0.25$, and $\bar n(0)=4$.

\subsection{Convergence trajectories}

Figure~\ref{fig:conv_traj} shows the trajectories $\bar x(t)$ for several control values $u$. 
For each control the system converges to a steady-state adherence level, while larger $u$ leads to lower equilibrium adherence. 
This behavior is consistent with \Cref{thm:global_convergence,lem:x_monotone}.

\begin{figure}[h]
\centering
\includegraphics[width=0.65\linewidth]{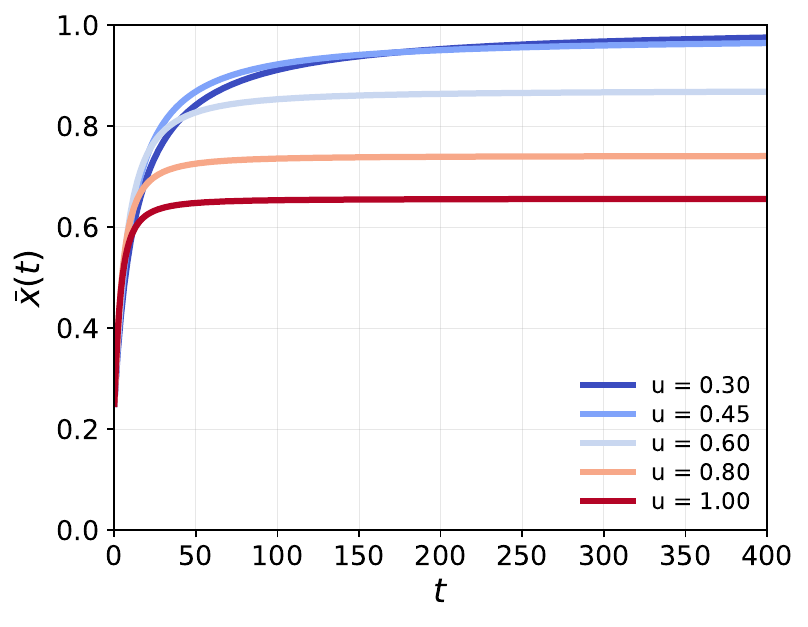}
\caption{Convergence trajectories of the adherence state $\bar x(t)$ for different constant controls $u$. 
Each trajectory converges to a control-dependent equilibrium.}
\label{fig:conv_traj}
\end{figure}

\subsection{Convergence behavior}

To illustrate the convergence rate of the mean-field recursion, we examine the error
\(e(t)=\bar x(t)-x^\star(u)\), where \(x^\star(u)\) denotes the equilibrium adherence level corresponding to control \(u\). Figure~\ref{fig:error_decay} shows the decay of \(|e(t)|\) on a log--log scale for several control values \(u\). 
The approximately linear decay on the log--log plot suggests a polynomial-type convergence behavior, which is consistent with the diminishing-step structure of the recursion and with the convergence established in \Cref{thm:global_convergence}.
The dashed line shows a reference curve of order \(C/t\).

\begin{figure}[h]
\centering
\includegraphics[width=0.65\linewidth]{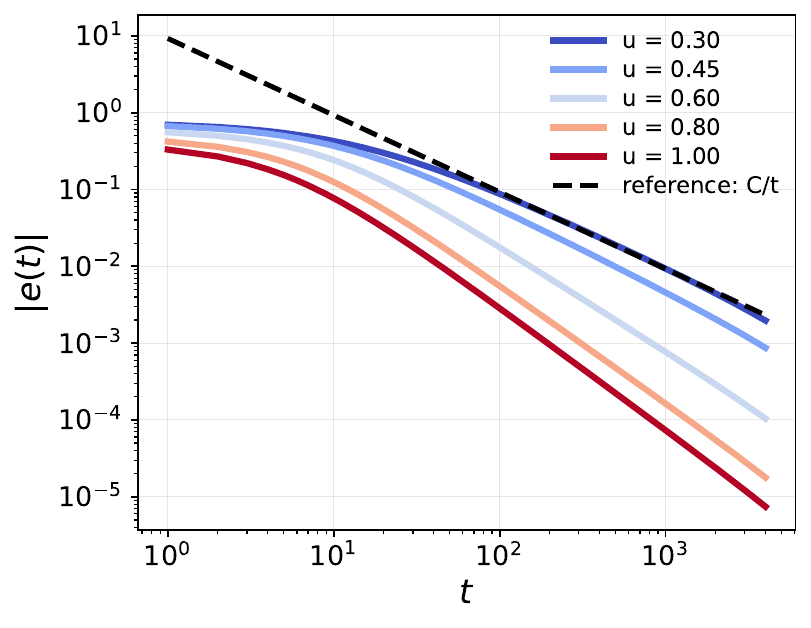}
\caption{Log--log decay of the adherence error \(|e(t)|\) for different controls $u$. 
The dashed line shows a reference curve of order $C/t$.}
\label{fig:error_decay}
\end{figure}

\subsection{Steady-state frontier}

For controls $u\in\{0.3,0.35,\dots,1\}$, the mean-field recursion is run for $T=1000$ steps and steady-state values are approximated by averaging the final $200$ samples. 
Figure~\ref{fig:frontier} plots the pairs $\bigl(x^\infty(u),\Gamma^\infty(u)\bigr)$, illustrating the adherence--throughput trade-off: larger $u$ increases throughput but reduces equilibrium adherence. 
With an adherence constraint $\underline x=0.9$, the optimal control is the largest feasible $u$, as predicted by \Cref{thm:const_opt_solution}.

\begin{figure}[h]
\centering
\includegraphics[width=0.65\linewidth]{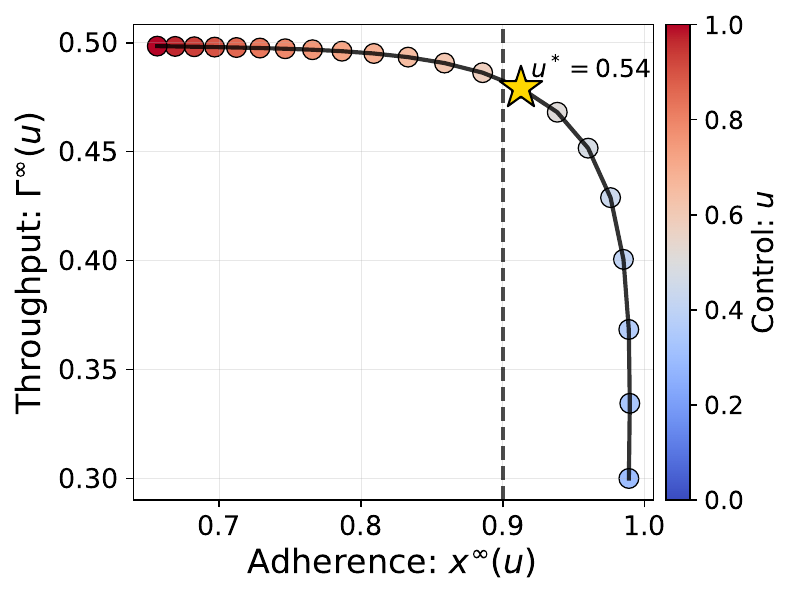}
\caption{Steady-state adherence--throughput frontier under constant control. 
Color indicates the control value $u$. 
The vertical line denotes the adherence constraint $\underline x=0.9$ and the star marks the optimal feasible control.}
\label{fig:frontier}
\end{figure}

\section{Conclusion}
We studied adherence-aware vehicle rebalancing under a mean-field control framework. 
Starting from a stochastic microscopic model of belief-driven participation, we derived a deterministic mean-field recursion for adherence dynamics and established well-posedness, equilibrium existence and uniqueness, and global convergence. 
The analysis reveals a structural trade-off between steady-state adherence and throughput under constant control. 
Exploiting this structure, we formulated and efficiently solved a control design problem that selects the largest recommendation intensity satisfying a prescribed adherence constraint.

Future work will study dynamic control policies and extend the framework to spatial mobility networks with data-driven parameter estimation.
\bibliographystyle{IEEEtran}
\bibliography{ref.bib}

\end{document}